\documentclass[11pt]{article}
\usepackage{makeidx}
\usepackage{color,array,graphics}
\usepackage{enumerate}
\usepackage[papersize={8.5in,11in},top=1in,left=1in,right=1in,bottom=1in]{geometry}

\usepackage{amsmath}
\usepackage{amsthm}
\usepackage{amssymb}
\usepackage{mathrsfs}
\usepackage{multirow}
\usepackage{tikz}
\usepackage{float}
\usepackage{cite}
\usepackage{bm}
\usepackage{hyperref}

\newtheorem{theorem}{Theorem}[section]
\newtheorem{corollary}{Corollary}[theorem]
\newtheorem{lemma}[theorem]{Lemma}

\usepackage[ruled]{algorithm2e}

\title{Ordinal Approximation for Social Choice, Matching,\\ and Facility Location Problems given Candidate Positions}

\author{Elliot Anshelevich and Wennan Zhu}

\begin{document}

\maketitle

\begin{abstract}
In this work we consider general facility location and social choice problems, in which sets of agents $\mathcal{A}$ and facilities $\mathcal{F}$ are located in a metric space, and our goal is to assign agents to facilities (as well as choose which facilities to open) in order to optimize the social cost. We form new algorithms to do this in the presence of only {\em ordinal information}, i.e., when the true costs or distances from the agents to the facilities are {\em unknown}, and only the ordinal preferences of the agents for the facilities are available. The main difference between our work and previous work in this area is that while we assume that only ordinal information about agent preferences in known, we know the exact locations of the possible facilities $\mathcal{F}$. Due to this extra information about the facilities, we are able to form powerful algorithms which have small {\em distortion}, i.e., perform almost as well as omniscient algorithms but use only ordinal information about agent preferences. For example, we present natural social choice mechanisms for choosing a single facility to open with distortion of at most 3 for minimizing both the total and the median social cost; this factor is provably the best possible. We analyze many general problems including matching, $k$-center, and $k$-median, and present black-box reductions from omniscient approximation algorithms with approximation factor $\beta$ to ordinal algorithms with approximation factor $1+2\beta$; doing this gives new ordinal algorithms for many important problems, and establishes a toolkit for analyzing such problems in the future.
\end{abstract}

\section{Introduction}

Many important problems involve assigning agents to facilities. For example, assigning patients to hospitals, students to universities, people to houses, etc. The target of assignment problems is usually to minimize social cost or maximize social welfare. When we consider the social cost of assignment problems, it is natural to assume the agents prefer facilities that are ``closer" to them in some sense, thus the social cost of an agent is often represented by the distance between the agent and the facility it is assigned to. Besides the cost of distances, there are many other cost functions and constraints for different problems; for example, in the capacitated facility assignment problem, each facility has a maximum number of agents it can accommodate. 

In this work we consider general facility location problems, in which sets of agents $\mathcal{A}$ and facilities $\mathcal{F}$ are located in a metric space, and our goal is to assign agents to facilities (as well as choose which facilities to open) so that agents are assigned to facilities which are close to them. For example, $\mathcal{F}$ may be possible locations for opening new stores, and the goal may be that all agents have a store near them, or that the sum of agent distances to the stores they are assigned to is small, etc. This setting also captures many social choice problems, in which the facilities correspond to candidates, and the goal would be to choose a single candidate (and assign all agents to this candidate) so that the distances from the agents to the chosen candidate are small. Here the distances correspond to {\em spatial preferences}, i.e., the metric space represents the ideological space in which a more preferred candidate would be closer to me; see \cite{enelow1984spatial,anshelevich2015approximating} for discussion of such spatial preferences in social choice. Our setting also captures matching and many related problems, in which we would open all facilities, but are only able to assign one agent to each facility, thus forming a matching between agents and facilities; facilities here could correspond to houses or items, for example.



If the distances between agents and facilities are known, then we can calculate the optimal solution for these assignment problems. Note that many of the facility location problems are NP-Complete, but at least it is possible to compute optimum assignments of agents to facilities (or the optimum candidates to select for social choice settings) given unlimited computational resources. For many of the settings we mentioned above, however, it is unlikely that we know the exact distances from the agents to the facilities. For social choice these distances would correspond to the cardinal preferences of voters for candidates, for example, ``My cost for candidate X winning is exactly 2.35." It is far more common that only {\em ordinal} preferences of the agents for the candidates are known, i.e., ``I prefer X to Y". Similarly, when trying to form a matching, or even in general facility location problems where we survey the agents to find out their preferences, it is much easier to elicit ordinal preferences (``I prefer to be matched with X over Y") over precise numerical preferences. These observations have recently led to a large body of work using the {\em utilitarian approach}, in which we assume that some latent numerical costs or utilities exist, but we only know the {\em ordinal} preferences of the agents, not their underlying numerical costs. See for example \cite{anshelevich2015approximating, anshelevich2015randomized, boutilier2015optimal, goel2016metric, feldman2016voting, skowron2017social, cheng2018distortion} for the social choice setting, \cite{ anshelevich2015blind, anshelevich2016truthful, anshelevich2017tradeoffs, abramowitz2017utilitarians} for matching and other graph problems, and \cite{caragiannis2016truthful} for facility location. These works focus on measuring the {\em distortion} of various algorithms: a measure of how well an algorithm behaves when using only ordinal information, as compared to the optimum algorithm which has access to the true underlying numerical information. More formally, the \textit{distortion} \cite{procaccia2006distortion,anshelevich2015approximating} of an assignment is defined as the worst-case ratio of its social cost to the social cost of the optimal solution.

As in the work mentioned above, we assume that only ordinal information about the distances between agents and facilities is known. However, although the locations and numerical preferences of the agents are usually difficult to obtain, the locations of facilities are mostly public information. The locations of political candidates in ideological space can be reasonably well estimated based on their voting records and public statements. When forming a survey about new stores to open, we may not know exactly how much the customers would prefer one store over the other since the customer locations may be private, but the locations of the possible stores themselves are public knowledge. The main difference between our work and previous work in this area is that we assume:\\

{\em \noindent While only ordinal information about agent preferences in known, we know the exact locations of the possible facilities $\mathcal{F}$.}\\

As we discuss below, this extra information about the locations of the facilities relative to each other allows us to produce much stronger algorithms, and show much nicer bounds on distortion. In fact, in many cases, we do not even need the full information about the locations of the facilities. The main message of this paper is that having a small amount of information about the candidates in social choice settings, or the facilities in facility location, allows us to obtain solutions which are provably {\em close to optimal} for a large class of problems even though the only information we have about the agent preferences is ordinal, and thus it is impossible (even given unlimited computational resources) to compute the {\em true} optimum solution.



\subsection{Our Contributions}
We begin by looking at the social choice setting, in which we have agents $\mathcal{A}$ and candidates $\mathcal{F}$ in a metric space, and we are given an ordinal ranking of each agent for the candidates. This setting was considered in e.g., \cite{anshelevich2015approximating, anshelevich2015randomized, goel2016metric, feldman2016voting, skowron2017social, cheng2018distortion, gross2017vote}. In particular, for the objective of minimizing the total distance from the agents to the chosen candidate, \cite{anshelevich2015approximating} showed that Copeland and similar voting mechanisms always have distortion of at most 5, while no deterministic voting mechanism can achieve a worst-case distortion of less than 3. Finding a deterministic mechanism with distortion less than 5 has been an open problem for several years \cite{goel2016metric}. In this paper, we show that if we know the exact locations of the candidates in addition to the ordinal ranking of the agents, then there is a simple algorithm which achieves a distortion of 3, and no better bound is possible. In other words, while we do not know the true distances from agents to candidates, we can compute an outcome which is a 3-approximation {\em no matter what} the true distances are, as long as they are consistent with the ordinal preferences given to us. Moreover, this approximation is possible even if for each agent we are only given their favorite (i.e., top-choice) candidate: there is no need for the agents to submit a full preference ranking over all the alternatives.

We also study other objective functions in addition to minimizing the total distance from agents to the chosen alternative. We give a natural deterministic voting mechanism which has distortion at most 3 for objectives such as minimizing the median voter cost, the egalitarian objective of minimizing maximum voter cost, and many other objectives. This mechanism achieves all these approximation guarantees {\em simultaneously}, and moreover it does not need the exact locations of the candidates: it suffices to be given an ordinal ranking of the distances from each candidate to each other candidate. In other words, this mechanism is especially suitable for the case when candidates are a subset of voters, as our mechanism will obtain the ordinal ranking of each voter for all the candidates, and this is the only information which would be required. Note that \cite{anshelevich2015approximating} proved that {\em no} deterministic mechanism can achieve a distortion of better than 5 for the median objective; the reason why we are able to achieve a distortion of 3 here is precisely because we also know how each candidate ranks all the other candidates, in addition to how each voter ranks all the candidates.

We then proceed to our general facility assignment model. We are given a set of agents and a set of facilities in a metric space. The distances between facilities are given, but the distances between agents and facilities are unknown; instead we only know ordinal preferences of agents over facilities which are consistent with the true underlying distances. There could be arbitrary constraints on the assignment, such as facility capacities, or constraints enforcing that some agents cannot be (or must be) assigned to the same facility, etc. A valid assignment is to assign each agent to a facility without violating the constraints. We consider many different social cost functions to optimize. For a general class of cost functions (essentially ones which are monotone and subadditive), we give a black-box reduction which converts an algorithm for the omniscient version of this problem (i.e., the version where the true distances are known) to an ordinal algorithm with small distortion. Specifically, if we have an omniscient algorithm which always produces an assignment which is a $\beta$-approximation to the optimum, then using it we can create an ordinal algorithm which only knows the ordinal preferences of the agents instead of their true distances to the facilities, but has distortion of at most $1+2\beta$.

%

  \begin{table}[htb]
  \centering
  \begin{tabular}{ | l | c | c | c | c | }
    \hline
    & Omniscient: & Agents' ordinal prefs & Only agents' ordinal \\
    & full distances & and facility locations & prefs (lower bounds)\\ \hline
    Total (Sum) Social Choice               & 1 & 3 & 5(3)\\ \hline
    Median Social Choice              & 1 & 3 & 5(5)\\ \hline
    Min Weight Bipartite Matching     & 1 & 3 & $n$(3)\\ \hline
    Egalitarian Bipartite Matching    & 1 & 3 & -(2)\\ \hline
    Facility Location                 & 1.488 \cite{li20111} & 3.976 & $\infty$ ($\infty$)  \\ \hline
    $k$-center                          & 2 \cite{hochbaum1985best}& 5 & - (-) \\ \hline
    $k$-median                          & 2.675 \cite{byrka2014improved}& 6.35 & - ($\Omega(n)$) \\ \hline
  \end{tabular}
  \caption{Best known distortion of polynomial-time algorithms in different settings. ``Omniscient'' stands for the setting where all the distances between agents and facilities are known, and the numbers represent the best-known approximation ratios. The second column represent our setting, in which the ordinal preferences of the agents, and the numerical distances between facilities are known. The last column represents the pure ordinal setting in which only the agent ordinal preferences are known, but the distances between facilities are unknown; this setting has been previously studied, and we include the known lower bounds on the possible distortion in parentheses, including some which we prove in the Appendix.}
   \label{table_results}
  \end{table}

Many well-known problems fall into our facility assignment model; Table \ref{table_results} summarizes some of our results. For example, classic facility location with facility costs, minimum weight bipartite matching, egalitarian bipartite matching, $k$-center, and $k$-median are all special cases. In particular our results show that if we are given unbounded computational resources, then it is always possible to form an assignment with distortion of at most 3 for these problems, and no better bound is possible simply due to the fact that we do not possess all the relevant information to compute the true optimum. This is a large improvement over previously known distortion bounds: for minimum cost ordinal matching the best-known distortion bound is $n$ using random serial dictatorship (RSD) \cite{caragiannis2016truthful}; by using the knowledge of facility locations we are able to reduce this approximation ratio to 3.

\subsection{Discussion and Related Work}
Ordinal approximation \cite{anshelevich2016ordinal} for the minimum social cost (or maximum social welfare) with underlying utilities/distances between agents and alternatives has been studied in many settings including social choice \cite{procaccia2006distortion, boutilier2015optimal, anshelevich2015approximating, anshelevich2015randomized, goel2016metric, feldman2016voting, caragiannis2017subset, skowron2017social, cheng2018distortion}, matchings \cite{bhalgat2011social, filos2014social, anshelevich2015blind, anshelevich2016truthful, caragiannis2016truthful, christodoulou2016social, anshelevich2017tradeoffs}, secretary problems \cite{hoefer2017combinatorial}, participatory budgeting \cite{benade2017preference}, general graph problems \cite{anshelevich2015blind, abramowitz2017utilitarians} and many other models in recent years. The general assumption of the ordinal setting is that we only have the ordinal preferences of agents over alternatives, and the goal is to form a solution that has close to optimal social cost. There are different models: social choice, matching, facility location, etc.; different objectives: minimizing social cost, maximizing social welfare, total cost objective, median objective, egalitarian objective, etc.; different assumptions on utility or cost functions: unit-sum, unit-range, metric space, etc. In this paper, we study general facility assignment problems in a metric space, and assume that the ordinal preferences of agents over alternatives are given. Unlike previous work on this topic, we also assume the locations of the alternatives are known; we show that this extra information enables us to achieve much better approximation ratios than in the pure ordinal setting for many problems.

The distortion of social choice functions was first introduced in \cite{procaccia2006distortion}, to describe the ratio between the total utility of the optimal candidate and the candidate selected by a mechanism using only ordinal preferences.
\cite{anshelevich2015approximating,skowron2017social,goel2016metric} studied the distortion of social choice functions in a metric space; the assumption that the underlying numerical costs have this metric property allows for much better results than more general costs. In particular, for the objective of minimizing the total distance from the agents to the chosen candidate, the above papers were able to show good distortion bounds for many well-known mechanisms, in particular a bound of 5 for Copeland \cite{anshelevich2015approximating}, a bound of $O(\ln m)$ for Single Transferable Vote (STV) \cite{skowron2017social}, and many others. In addition, \cite{anshelevich2015approximating} proved that no deterministic mechanism can have worst-case distortion better than 3, and \cite{skowron2017social} showed that all scoring rules for $m$-candidates have a distortion of at least $1 + 2 \sqrt{\ln m - 1}$. Goel et al. \cite{goel2016metric} showed that Ranked Pairs, and the Schulze rule have a worst-case distortion of at least 5, and the expected worst-case distortion of any (weighted)-tournament rule is at least 3. They also introduced the notion of ``fairness'' of social choice rules, and discussed the fairness ratio of Copeland, Randomized Dictatorship, and a general class of cost functions. Finding a deterministic mechanism with distortion less than 5 has been an open problem for several years. In this paper, we show that if we know the exact locations of the candidates in addition to the ordinal ranking of the agents, then there is a simple algorithm which achieves a distortion of 3, and no better bound is possible.

While the above work, as well as our paper, only focuses on deterministic algorithms, the distortion of randomized algorithms in social choice has also been considered, see for example \cite{anshelevich2015randomized,fain2017sequential,gross2017vote,feldman2016voting}. In a slightly different flavor of result, \cite{cheng2017people,cheng2018distortion} consider the special case where candidates are randomly and independently drawn from the set of voters. While we leave the analysis of randomized algorithms which know the location of the facilities to future work, and consider the worst-case candidate locations, it is worth pointing out that our {\em deterministic} algorithm achieves a distortion of 3, which is also the best known distortion bound for any {\em randomized} mechanism which only knows the ordinal preferences of the agents. Similarly, another common goal is to form {\em truthful} mechanisms with small distortion for matching and social choice, as in \cite{feldman2016voting,anshelevich2016truthful, caragiannis2016truthful}; we focus on general mechanisms in this paper in order to understand the limitations of knowing only certain kinds of ordinal information, and leave the goal of forming truthful mechanisms for future work.

For the median objective of social choice problems, \cite{anshelevich2015approximating} showed that Copeland gives a distortion of at most 5, while {\em no} deterministic mechanism can achieve a distortion of better than 5 . \cite{anshelevich2015randomized} also gave a randomized algorithm that has a distortion of at most 4. In this paper, we are able to improve this bound to a tight worst-case distortion of 3 by a deterministic mechanism, because we also know how each candidate ranks all the other candidates, in addition to how each voter ranks all the candidates.

The distortion of matching in a metric space has received far less attention than social choice questions. \cite{anshelevich2015blind, anshelevich2016truthful, anshelevich2017tradeoffs} analyzed maximum-weight metric matching; the maximization objective makes this problem far easier, and even choosing a uniformly random matching yields a distortion of a small constant. This is very different from our goal of computing a minimum-cost matching, for which no ordinal approximations better than $O(n)$ are known. \cite{caragiannis2016truthful} studied facility assignment problems in a metric space; they considered the problem with or without resource augmentation, and the cases without augmentation are exactly the minimum weight bipartite matching problem. \cite{caragiannis2016truthful} showed that the approximation ratio of random serial dictatorship (RSD) is at most $n$, and gave a lower bound of $2^n - 1$ for the approximation ratio of serial dictatorship (SD), and a lower bound of $n^{0.26}$ for RSD. Their results are the best known ordinal approximations for this problem. In this paper, we are able to give a tight 3-approximation for the minimum weight matching problem, given the locations of facilities in addition to the agents' ordinal preferences.


\section{Model and Notation: Social Choice}
\label{section-social-model}

For the social choice problems studied in this paper, we let $\mathcal{A} = \{ 1, 2, \dots, n\}$ be a set of agents, and let $\mathcal{F} = \{ F_1, F_2, \dots, F_m\}$ be a set of alternatives, which we will also refer to sometimes as {\em candidates} or {\em facilities}. We will typically use $i$ and $j$ to refer to agents and $W, X, Y, Z$ to refer to alternatives. Let $\mathcal{S}$ be the set of total orders on the set of alternatives $\mathcal{F}$. Every agent $i \in \mathcal{A}$ has a preference ranking $\sigma \in \mathcal{S}$; by $X \succ_i Y$ we will mean that $X$ is preferred over $Y$ in ranking $\sigma$. Although we assume that each agent has a total order of preference over the alternatives and that this order is known to us, for many of our results it is only necessary that the top choice of each agent is known. We say $X$ is $i$'s top choice if $i$ prefers $X$ to every other alternative in $\mathcal{F}$. We call the vector $\sigma = (\sigma_1, \dots, \sigma_n) \in \mathcal{S}^n$ a preference profile. We say that an alternative $X$ pairwise defeats $Y$ if $| \{i \in \mathcal{A}: X \succ_i Y \}| > \frac{n}{2}$. The goal is to choose a single winning alternative.


\paragraph{Cardinal Metric Costs.} In this work we take the utilitarian view, and assume that the ordinal preferences $\sigma$ are derived from underlying (latent) cardinal agent costs. Formally, we assume that there exists an arbitrary metric $d: (\mathcal{A} \cup \mathcal{F})^2 \rightarrow \mathbb{R}_{\ge 0}$ on the set of agents and alternatives. The cost incurred by agent $i$ of alternative $X$ being selected is represented by $d(i, X)$, which is the distance between $i$ and $X$. Such spatial preferences are relatively common and well-motivated, see for example \cite{enelow1984spatial,anshelevich2015approximating} and the references therein. The underlying distances $d(i, X)$ are {\em unknown}, but unlike most previous work we {\em do} assume the distances between alternatives are given. For example, when alternatives represent facilities or stores to be opened, it makes sense that their specific locations would be known, while the distances from the customers to the stores may be private. Similarly, when the alternatives represent political candidates, it may be easy to estimate their locations in ideological space (for example based on their voting records and public statements), but the ideology of the voters is much harder to estimate, with mechanism designers only knowing which candidates the voters prefer but not how much they prefer them. The distance between two alternatives $X$ and $Y$ is denoted by $l(X, Y)$. We say that $d$ is {\em consistent} with $l$ if $\forall X,\ Y \in \mathcal{F}$, $d(X, Y) = l(X, Y)$. 

The metric costs $d$ naturally give rise to a preference profile. We say that $d$ is {\em consistent} with $\sigma$ if $\forall i \in \mathcal{A}$, $\forall X,\ Y \in \mathcal{F}$,  if $d(i, X) < d(i, Y)$, then $X \succ_i Y$. It means that the cost of $X$ is less than the cost of $Y$ for agent $i$, so agent $i$ prefers $X$ over $Y$. As described above, we know exactly the distances $l$ and the preferences $\sigma$, but do not know the true costs $d$ which give rise to $\sigma$. Let $\mathcal{D}(\sigma, l)$ be the set of metrics that are consistent with $\sigma$ and $l$; we know that one of the metrics from this possibly infinite space captures the true costs, but do not know which one.

\paragraph{Social Cost Distortion}
We study several objective functions for social cost in this paper. First, the most common notion of social cost is the sum objective function, defined as $SC_{\sum}(X, \mathcal{A}) = \sum_{i \in \mathcal{A}} d(i, X)$. We also study the median objective function, $SC_{\text{med}}(X, \mathcal{A}) = \text{med}_{i \in \mathcal{A}} (d(i, X))$, as well as the egalitarian objective and many others (see Section \ref{subsection-social-median}).  We use the notion of distortion to quantify the quality of an alternative in the worst case, similar to the notation in \cite{boutilier2015optimal, procaccia2006distortion}. For any alternative $W$, we define the distortion of $W$ as the ratio between the social cost of $W$ and the optimal alternative:

\begin{align*}
  dist_{\sum}(W, \sigma, l) &= \sup\limits_{d \in \mathcal{D}(\sigma, l)} \frac{SC_{\sum}(W, \mathcal{A})}{\min_{X \in \mathcal{F}}SC_{\sum}(X, \mathcal{A})} \\
  dist_{\text{med}}(W, \sigma, l) &= \sup\limits_{d \in \mathcal{D}(\sigma, l)} \frac{SC_{\text{med}}(W, \mathcal{A})}{\min_{X \in \mathcal{F}}SC_{\text{med}}(X, \mathcal{A})}
\end{align*}

In other words, saying that the distortion of $W$ is at most 3 means that, no matter what the true costs $d$ are (as long as they are consistent with the $\sigma$ and $l$ which we know), it must be that the social cost of $W$ is within a factor of 3 of the true optimum alternative, which is impossible to compute without knowing the true costs. Because of this, a small distortion value means that there is no need to obtain the true agent costs, and the ordinal information $\sigma$ (together with information $l$ about the alternatives) is enough to form a good solution.

A social choice function $f$ on $\mathcal{A}$ and $\mathcal{F}$ takes $\sigma$ and $l$ as input, and returns the winning alternative. We say the distortion of $f$ is the same as the distortion of the winning alternative chosen by $f$ on $\sigma$ and $l$. In other words, the distortion of a social choice mechanism $f$ on a profile $\sigma$ and facility distances $l$ is the worst-case ratio between the social cost of $W = f(\sigma, l)$, and the social cost of the true optimal alternative.\\

%

\section{Distortion of Social Choice Mechanisms}
\label{section-social-choice}

\subsection{Distortion of Total Social Cost}
\label{subsection-social-sum}
In this section, we study the sum objective and provide a deterministic algorithm that gives a distortion of at most 3. According to \cite{anshelevich2015approximating}, the lower bound on the distortion for deterministic social choice functions with only ordinal preferences (without knowing $l$) is 3. This occurs in the simple example with 2 alternatives which are tied with approximately half preferring each one. No matter which one is chosen, the true optimum could be the other one, and its social cost can be as much as 3 times better. Because the example in Theorem 3 from \cite{anshelevich2015approximating} only has two alternatives, knowing $l$ does not provide any extra information, and thus that example also provides a lower bound of 3 in our setting, although we assume the distances $l$ between facilities are known in this paper. Therefore, our mechanism achieves the best possible distortion in this setting. Note that if we only have ordinal preferences of the agents without the distances between facilities, then the best known approach so far is Copeland, which gives a distortion at most 5. Thus our results establish that by knowing the distances $l$ between alternatives, it is possible to reduce the distortion from 5 to 3, and no better deterministic mechanism is possible.

\begin{lemma}
  \label{lemma-social-basic}
  Let $W,\ X$ be alternatives. If $W\succ_i X$, then $d(i, X) \ge \frac{d(X, W)}{2}$. [Lemma 5 in \cite{anshelevich2015approximating}]
\end{lemma}

In the following algorithm, we generate a set of projected agents as follows: Given agents $\mathcal{A}$, alternatives $\mathcal{F}$, and the preference profile $\sigma$, for each agent $i$ denote alternative $X_i$ as $i$'s top choice. Then we create a new agent $\tilde{i}$ at the location of $X_i$ in the metric space, as shown in Figure~\ref{fig:projected} (a); consequently, $\forall \ Y \in \mathcal{F}$, $d(\tilde{i}, Y) = d(X_i, Y)$. Denote the set of the new agents as $\tilde{\mathcal{A}} = \{ \tilde{1}, \tilde{2}, \dots, \tilde{n} \}$. For any metric $d$ consistent with $l$, $d(\tilde{i}, Y) = d(X_i, Y) = l(X_i, Y)$, so the distances between agents in $\tilde{\mathcal{A}}$ and alternatives in $\mathcal{F}$ are known to us, unlike the true distances between $\mathcal{A}$ and $\mathcal{F}$. 

\begin{figure}[htb]
\begin{center}
\includegraphics[scale=0.7]{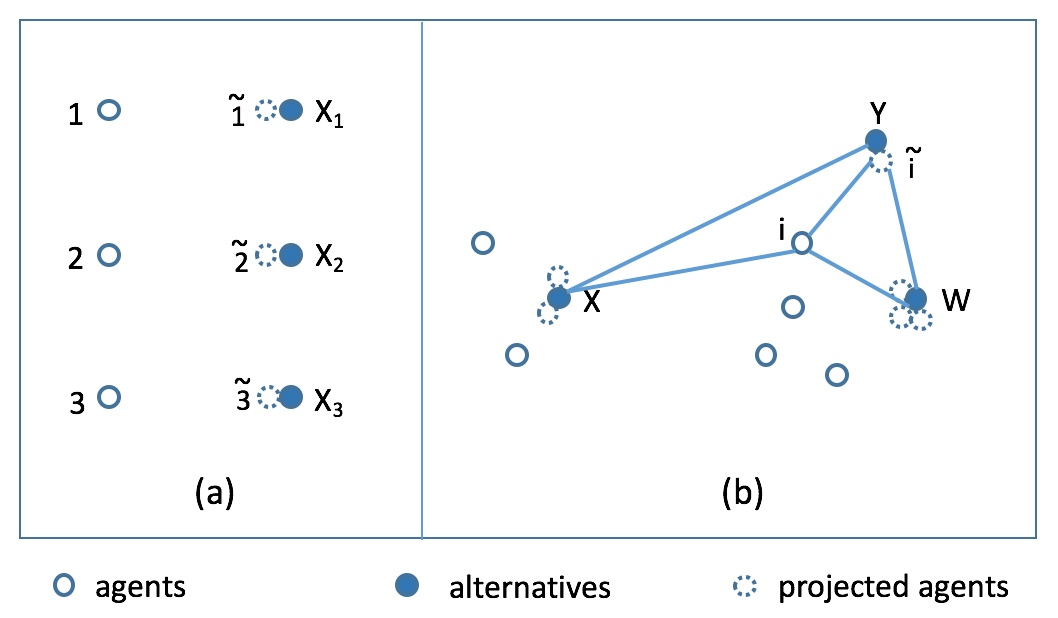}
\end{center}
\caption{(a) For each agent, generate a projected agent at the location of its top choice alternative. (b) A figure demonstrating agent $i$, $i$'s top choice alternative $Y$, $i$'s projected agent $\tilde{i}$ located at $Y$, the winner $W$, and the optimal alternative $X$ for the proof of Theorem~\ref{thm-social-sum}.}
\label{fig:projected}
\end{figure}

\begin{algorithm}[htb]
\caption{Algorithm for the minimum total social cost.}
\label{alg-social-sum}
 \SetKwInOut{Input}{Input}
 \SetKwInOut{Output}{Output}
 \Input{Agents $\mathcal{A} = \{ 1, 2, \dots, n\}$, \\
        Alternatives $\mathcal{F} = \{ F_1, F_2, \dots, F_m\}$,\\
        Each agent $i$'s top choice alternative,\\
        Distances between alternatives, i.e., $l(Y, Z)$, $\forall \ Y, Z \in \mathcal{F}$
        }
 \Output{ The winning alternative $W$. }

 Generate projected agent set $\tilde{\mathcal{A}}$. For each alternative $X \in \mathcal{F}$, calculate the total social cost on $\tilde{\mathcal{A}}$ by choosing $X$, i.e., $SC_{\sum}(X, \tilde{\mathcal{A}}) = \sum_{\tilde{i} \in \tilde{\mathcal{A}} } d(\tilde{i}, X) = \sum_{\tilde{i} \in \tilde{\mathcal{A}} } l(\tilde{i}, X)$  .

 \textbf{Final Output:} Return the alternative $W$ that has the minimum social cost $SC_{\sum}(W, \tilde{\mathcal{A}})$ .
\end{algorithm}

\begin{theorem}
\label{thm-social-sum}
The distortion of Algorithm~\ref{alg-social-sum} for minimum total social cost on $\mathcal{A}$ is at most 3.
\end{theorem}

\begin{proof}
  Let $W$ denote the winning alternative. $W$ has the minimum social cost on the agent set $\tilde{\mathcal{A}}$, so for any alternative $Y$, it must be that

  \begin{equation}
    \frac{SC_{\sum}(W, \tilde{\mathcal{A}})}{SC_{\sum}(Y, \tilde{\mathcal{A}})} = \frac{\sum_{\tilde{i} \in \tilde{\mathcal{A}} } d(\tilde{i}, W)}{\sum_{\tilde{i} \in \tilde{\mathcal{A}} } d(\tilde{i}, Y)} =  \frac{\sum_{i \in \mathcal{A} } d(\tilde{i}, W)}{\sum_{i \in \mathcal{A} } d(\tilde{i}, Y)} \le 1 \label{lemma-loc-eq1}
  \end{equation}

  Let $X$ denote the true optimal alternative for $\mathcal{A}$. 
%
  We want to get $dist_{\sum}(W, \sigma, l)$ by upper bounding the cost incurred by $W$ compared to $X$:

  \begin{align}
    \frac{SC_{\sum}(W, \mathcal{A})}{SC_{\sum}(X, \mathcal{A})} &= \frac{\sum_{i \in \mathcal{A}} d(i, W)}{\sum_{i \in \mathcal{A}} d(i, X)} \nonumber \\
    &\le \frac{\sum_{i \in \mathcal{A}} (d(i, \tilde{i}) + d(\tilde{i}, W))}{\sum_{i \in \mathcal{A}} d(i, X)} \nonumber \\
    &= \frac{\sum_{i \in \mathcal{A}} d(i, \tilde{i})}{\sum_{i \in \mathcal{A}} d(i, X)} + \frac{\sum_{i \in \mathcal{A}} d(\tilde{i}, W)}{\sum_{i \in \mathcal{A}} d(i, X)} \label{lemma-loc-eq2}
  \end{align}

  The inequality $d(i, W) \le d(i, \tilde{i}) + d(\tilde{i}, W)$ is due to the triangle inequality since $d$ is a metric, as shown in Figure~\ref{fig:projected} (b). $\forall i \in \mathcal{A}$, we know that $\tilde{i}$ is located at $i$'s top choice alternative, so the distance between $i$ and $\tilde{i}$ must be less than (or equal to) the distance between $i$ and any alternative; thus $d(i, \tilde{i}) \le d(i, X)$. Summing up for all $i \in \mathcal{A}$, we get that $\frac{\sum_{i \in \mathcal{A}} d(i, \tilde{i})}{\sum_{i \in \mathcal{A}} d(i, X)} \le 1$. For any agent $i$ such that $X$ is not $i$'s top choice, suppose alternative $Y$ is $i$'s top choice, then $\tilde{i}$ has the same location as $Y$ and $d(\tilde{i}, X) = d(X, Y)$. By Lemma~\ref{lemma-social-basic}, $d(i, X) \ge \frac{d(X, Y)}{2}$, thus $d(i, X) \ge \frac{d(\tilde{i}, X)}{2}$. For all $i$ that $X$ is $i$'s top choice, $d(\tilde{i}, X) = 0$, so the inequality $d(i, X) \ge \frac{d(\tilde{i}, X)}{2}$ holds for all $i \in \mathcal{A}$. Together with inequality~\ref{lemma-loc-eq1} and ~\ref{lemma-loc-eq2},

  \begin{equation*}
    \frac{SC_{\sum}(W, \mathcal{A})}{SC_{\sum}(X, \mathcal{A})} \le 1 + \frac{\sum_{i \in \mathcal{A}} d(\tilde{i}, W)}{\sum_{i \in \mathcal{A}} \frac{d(\tilde{i}, X)}{2}}
    = 1 + 2\frac{\sum_{i \in \mathcal{A}} d(\tilde{i}, W)}{\sum_{i \in \mathcal{A}} d(\tilde{i}, X)}
    \le 3
  \end{equation*}

\end{proof}

\subsection{Distortion of Median Social Cost}
\label{subsection-social-median}
In this section, we study the median objective function, and provide a deterministic mechanism that gives a distortion of at most 3. Recall that we define the median social cost of an alternative $X$ as $SC_{\text{med}}(X, \mathcal{A}) = \text{med}_{i \in \mathcal{A}} (d(i, X))$. We will refer to this as $\text{med}(X)$ when $d$ and $\mathcal{A}$ are fixed. If $n$ is even, we define median to be the $(\frac{n}{2} + 1)^{th}$ smallest value of the distances. Note that no deterministic mechanism which only knows ordinal preferences can have worst-case distortion better than 5 (Theorem 14 in \cite{anshelevich2015approximating}). With known distances between facilities, we are able to provide a natural social choice function with distortion of 3, which is also provably the best possible distortion in our setting (consider the example in Theorem 3 from \cite{anshelevich2015approximating} again). Moreover, our social choice function only uses ordinal information about the alternatives, and not the full distances $l$; in particular as long as we have ordinal preferences of each alternative for each other alternative (and thus a total order of the distances from each alternative to the others), then our mechanism will work properly. Such ordinal information may be easier to obtain than full distances $l$; for example candidates can rank all the other candidates. In particular, given agents with ordinal preferences such that the candidates are a subset of the agents, our mechanism will always form an outcome with small distortion, even if we do not know the distances $l$.

Note that using only agents' top choices over alternatives and the distances between alternatives, as Algorithm~\ref{alg-social-sum} does for the total social cost objective, is not enough to give a worst-case distortion of 3 for the median objective. Consider the following example: there are 4 alternatives $W, X, Y, Z$, the distances between them are: $d(W, Y) = d(Y, X) = d(X, Z) = d(Z, W) = 2$ and $d(W, X) = d(Y,Z) = 4$. Suppose $W$ is agents 1, 2's top choice, $X$ is agent 3, 4's top choice, $Y$ is agent 5, 6's top choice, and $Z$ is agent 7, 8's top choice. This graph is symmetric, so we choose an arbitrary alternative as the winner. Suppose we choose $W$ as the winner, and the distances between agents and facilities are: the distances from agents 1, 2 to $W$ are both 100, the distances from agents 1, 2 to $X, Y, Z$ are all 102. The distances from agents 5, 6 to $Y, X$ are all 1, and the distances from agents 5, 6 to $W, Z$ are all 3. The distances from agents 7, 8 to $Z, X$ are all 1, and the distances from agents 7, 8 to $Y, W$ are all 3. The distances from agents 3, 4 to $X$ are both 1, the distances from 3, 4 to $Y, Z$ are all 3, and the distances from 3, 4 to $W$ are both 5. In this example, the median is the distance from $5^{th}$ closest agent to the winning alternative. $X$ is the optimal alternative with $\text{med}(W) = 1$, while $\text{med}(W) = 5$ has a distortion of 5.\\

We will use the following Lemmas from \cite{anshelevich2015approximating} in the proof of our algorithm:

\begin{lemma}
\label{lemma-social-median1}
  For any two alternatives $W$ and $Y$, we have $\text{med}(W) \le \text{med}(Y) + d(Y, W)$. [Lemma 11 in \cite{anshelevich2015approximating}]
\end{lemma}

\begin{lemma}
\label{lemma-social-median2}
  For any two alternatives $Y$ and $P$, if $P$ pairwise defeats (or pairwise ties) $Y$, then $\text{med}(Y) \ge \frac{d(Y, P)}{2}$. [Proved in Theorem 16 in \cite{anshelevich2015approximating}]
\end{lemma}


\begin{lemma}
\label{lemma-social-median3}
Let $W, Y$ be an alternatives $\in \mathcal{F}$, if $W$ pairwise defeats (or pairwise ties) $Y$, then $\text{med}(W) \le 3\text{med}(Y)$. [Proved in Theorem 8 in \cite{anshelevich2015approximating}]
\end{lemma}

The main easy insight which we use in the formation of our algorithm comes from the following lemma.

\begin{lemma}
\label{lemma-social-median4}
  For any three alternatives $W$, $Y$, and $P$, if $P$ pairwise defeats (or pairwise ties) $Y$, and $d(Y, W) \le d(Y, P)$, then $\text{med}(W) \le 3 \text{med}(Y)$.
\end{lemma}

\begin {proof}
  By Lemma~\ref{lemma-social-median1}, $\text{med}(W) \le \text{med}(Y) + d(Y, W)$. By Lemma~\ref{lemma-social-median2}, $\text{med}(Y) \ge \frac{d(Y, P)}{2}$. And we know that $d(Y, P) \ge d(Y, W)$, thus

  \begin{align*}
      \text{med}(W) &\le \text{med}(Y) + d(Y, W)\\
                    &\le \text{med}(Y) + d(Y, P)\\
                    &\le \text{med}(Y) + 2\text{med}(Y)\\
                    &\le 3\text{med}(Y)
  \end{align*}
\end{proof}

  We use a natural Condorcet-consistent algorithm to approximate the minimum median social cost with the agents' preference rankings $\sigma$ and the ordinal preferences of every alternative over other alternatives. First, create the majority graph $G=(\mathcal{F}, E)$, i.e., a graph with alternatives as vertices and an edge $(Y, Z) \in E$ if $Y$ pairwise defeats or pairwise ties $Z$. If a Condorcet winner (i.e. an alternative which pairwise defeats all others) exists, then we return it immediately.

  Otherwise, we consider each pair of alternatives. By Lemma~\ref{lemma-social-median3}, if the edge $(W, Y) \in E$, then $\text{med}(W) \le 3\text{med}(Y)$. When considering an alternative pair $W, Y$, if $(W, Y) \not\in E$ and we know that there exists another alternative $P$ which meets the conditions of Lemma~\ref{lemma-social-median4}, then we add an edge $(W, Y)$ to $G$. It is not difficult to see that whenever the edge $(W, Y)$ is in our graph, this means that $\text{med}(W) \le 3\text{med}(Y)$. As we prove below, at the end of this process there always exists at least one alternative which has edges to all the other alternatives, and thus the distortion obtained from selecting it is at most 3, no matter which alternative is the true optimal one.


  Note that from the ordinal preferences of alternatives over each other, we can get a partial order of distances between the alternatives. Denote this partial order as $\preceq$, i.e., we say that $d(W,Y)\preceq d(W,Z)$ if we know that $W$ prefers $Y$ to $Z$ (we do not have information about strict preference). This is the information we have on hand: we only know the partial order of distances between pairs of alternatives which share an alternative in common. Note, however, that if there exists a cycle in this partial order, i.e., $d(Y_1,Y_2) \preceq d(Y_2,Y_3) \preceq d(Y_3,Y_4) \preceq \dots \preceq d(Y_k,Y_1) \preceq d(Y_1,Y_2)$, then this implies that all the distances in the cycle are actually equal, and thus we can also add the relations $d(Y_1,Y_2) \succeq d(Y_2,Y_3) \succeq d(Y_3,Y_4) \succeq \dots \succeq d(Y_k,Y_1) \succeq d(Y_1,Y_2)$. Such cycles are easy to detect (e.g., by forming a graph with a node for every alternative pair and then searching for cycles), and thus we can assume that whenever a cycle exists in our partial order, then for every pair of distances $d(W,Y)$ and $d(W,Z)$ in the cycle, we have both $d(W,Y)\preceq d(W,Z)$ and $d(W,Y)\succeq d(W,Z)$.


\begin{algorithm}[htb]
\caption{Algorithm for the minimum median social cost.}
\label{alg-social-median}
 \SetKwInOut{Input}{Input}
 \SetKwInOut{Output}{Output}
 \Input{Agents $\mathcal{A} = \{ 1, 2, \dots, n\}$, \\
        Alternatives $\mathcal{F} = \{ F_1, F_2, \dots, F_m\}$,\\
        The majority graph $G=(\mathcal{F}, E)$,\\
        Ordinal preferences of each alternative over other alternatives,\\
        Partial order of distances between alternatives.
        }
 \Output{ The winning alternative $W$. }

 \vspace{2mm}
 If there is a Condorcet winner $W$, \textbf{return $W$ as the winner}.

 \vspace{2mm}

 \ForAll {alternative pairs $W, Y$} {
   \If{$(W, Y) \not\in E$ or $(Y, W) \not\in E$ } {
    WLOG, suppose $(Y, W)$ exists, but $(W, Y)$ does not exist.\\
    \If{there exists an alternative $P$, such that we have $d(Y, W) \preceq d(Y, P)$ in our partial order information, and $P$ pairwise defeats (or ties) $Y$}{
      Add edge $(W, Y)$ to $E$\;
      continue\;
    }
  }
}
There must exists an alternative $W$ such that $\forall Y \in \mathcal{F} - \{ W \}$, $(W, Y) \in E$. Return $W$ as the winner.
\end{algorithm}

\begin{lemma}
\label{lemma-social-median5}
  Consider the modified majority graph $G=(\mathcal{F}, E)$ at any point during Algorithm~\ref{alg-social-median}. For any edge $(W, Y) \in E$, we have that $\text{med}(W) \le 3\text{med}(Y)$.
\end{lemma}

\begin{proof}

  By Lemma~\ref{lemma-social-median3}, for any edge $(W, Y)$ in the original majority graph, $\text{med}(W) \le 3\text{med}(Y)$.

  Now consider an edge $(W, Y)$ added to $E$ when processing the alternative pair $W, Y$. It must be the case that there exists an alternative $P$, such that $d(Y, W) \le d(Y, P)$ and $P$ pairwise defeats (or ties) $Y$. By Lemma~\ref{lemma-social-median4}, $\text{med}(W) \le 3\text{med}(Y)$.
\end{proof}

\begin{lemma}
\label{lemma-social-median6}
At the end of Algorithm~\ref{alg-social-median}, there must exist an alternative $W$ such that $\forall Y \in \mathcal{F} - \{ W \}$, $(W, Y) \in E$.
\end{lemma}

\begin{proof}
  We prove this lemma by contradiction. Suppose no such alternative $W$ exists. Then for each alternative $Y$, there is at least one alternative $Z$, such that only $(Z, Y) \in E$ and $(Y, Z) \not\in E$. This is because we start with the majority graph, so at least one edge always exists between every pair. We create another directed graph $G' = (\mathcal{F}, E')$, with $E'$ being all the edges $(Z,Y)$ such that $Y,Z\not\in E$. Thus any pair of alternatives in $G'$ have at most one direction of edge between them. And by our assumption, each alternative $Y$ has at least one incoming edge in $G'$. Since the in-degree of each node is at least 1 in $G'$, there must be at least one cycle in $G'$. To see this, one can for example take the edge $(Y_2,Y_1)$ coming into $Y_1$, then the edge $(Y_3,Y_2)$ coming into $Y_2$, and proceed in this way until a cycle is formed. Note that every edge in $G'$ must be in the original majority graph, because if we add an edge when processing a pair of alternatives in our algorithm, that pair must have edges in both directions.

  Consider a cycle formed by edges $(Y_1, Y_2)$, $(Y_2, Y_3)$, \dots, $(Y_{k-1}, Y_k)$, $(Y_k, Y_1)$. When processing the alternative pair $Y_2, Y_3$ in Algorithm~\ref{alg-social-median}, we did not add edge $(Y_3, Y_2)$ to $E$, so it must be the case that no alternative $P$ exists such that $d(Y_2, Y_3) \preceq d(Y_2, P)$ and $P$ pairwise defeats (or ties) $Y_2$. But we know that $Y_1$ pairwise defeats (or ties) $Y_2$, because edge $(Y_1, Y_2)$ is in the original majority graph. Then the only possibility is we don't know if $d(Y_2, Y_3) \le d(Y_1, Y_2)$, i.e., either $d(Y_2, Y_3)$ and $d(Y_1, Y_2)$ are incomparable in our partial order, or we only know that $d(Y_2, Y_3) \succeq d(Y_1, Y_2)$. They cannot be incomparable, since we have the ordinal preferences of $Y_2$ for $Y_1$ and $Y_3$, thus our partial order must state that $d(Y_2, Y_3) \succeq d(Y_1, Y_2)$, i.e., $Y_2$ prefers $Y_1$ to $Y_3$.
  By the same reasoning, we also get that $Y_3$ prefers $Y_2$ to $Y_4$, and more generally that $Y_i$ prefers $Y_{i-1}$ to $Y_{i+1}$ for all $i$, where $Y_0=Y_k$ and $Y_{k+1}=Y_1$ since it is a cycle. This means that in our partial order, we have that $d(Y_1, Y_2) \preceq d(Y_2, Y_3) \preceq \dots \preceq d(Y_{k-1}, Y_k) \preceq d(Y_k, Y_1) \preceq d(Y_1, Y_2)$. Recall, however, that this means we know $d(Y_1, Y_2) = d(Y_2, Y_3) = \dots = d(Y_{k-1}, Y_k) = d(Y_k, Y_1)$, and before running Algorithm~\ref{alg-social-median}, we detect cycles in the partial order of alternative distances, and add the equality information to the partial order. This means that whenever $d(Y_1, Y_2) \preceq d(Y_2, Y_3) \preceq \dots \preceq d(Y_{k-1}, Y_k) \preceq d(Y_k, Y_1) \preceq d(Y_1, Y_2)$ exists in our partial order, we also have $d(Y_1, Y_2) \succeq d(Y_2, Y_3) \succeq \dots \succeq d(Y_{k-1}, Y_k) \succeq d(Y_k, Y_1) \succeq d(Y_1, Y_2)$ in the partial order as well. But this gives us a contradiction, since having $d(Y_2, Y_3) \preceq d(Y_1, Y_2)$ in the partial order, combined with the fact that $Y_1$ pairwise defeats $Y_2$, would cause us to add the edge $(Y_3,Y_2)$ in our algorithm, which contradicts the statement that only the edge $(Y_2,Y_3)$ is in the final graph produced by the algorithm, but not $(Y_3,Y_2)$. Thus there must exist at least one alternative with edges from it to all the others.
\end{proof}

\begin{theorem}
\label{thm-social-median}
The distortion of Algorithm~\ref{alg-social-median} for minimum median social cost is at most 3.
\end{theorem}

\begin{proof}
  If there is a Condorcet winner, by Lemma~\ref{lemma-social-median3}, the distortion is at most 3.

  Otherwise, by Lemma~\ref{lemma-social-median6}, the algorithm always returns a winner. Suppose it returns alternative $W$ as the winner, by Lemma~\ref{lemma-social-median5}, $W$ has a distortion at most 3 with any alternative $X$ as the optimal solution.
\end{proof}

\subsubsection{Generalizing Median: Percentile Distortion}
Instead of just considering the median objective, we also consider a more general objective: the $\alpha$-percentile social cost. Let $\alpha\text{-PC}(Y)$ denote the value from the set $\{ d(i, Y) : i \in \mathcal{A} \}$, that $\alpha$ fraction of the values lie below $\alpha\text{-PC}(Y)$. Thus median is a special case when $\alpha=\frac{1}{2}$, $\text{med}(Y) = \frac{1}{2}\text{-PC}(Y)$. It was shown in \cite{anshelevich2015approximating} Theorem 17 that the worst-case distortion when $\alpha \in [0, \frac{1}{2}]$ in that setting (only have agent's ordinal preferences over alternatives) is unbounded, and the same example shows $\alpha \in [0, \frac{1}{2}]$ in our setting is also unbounded. However, we are able to give a distortion of 3 for $\alpha \in [\frac{1}{2}, 1]$ in this paper, while for the setting in \cite{anshelevich2015approximating}, the lower bound for distortion when $\alpha \in [\frac{1}{2}, \frac{2}{3}]$ is 5. The reason is that the ordinal preferences between alternatives are also available in our setting. We will show that Algorithm~\ref{alg-social-median} gives a distortion of at most 3 not only for the median objective, but also for the general $\alpha$-percentile objective, because all the lemmas we used to prove the conclusion for the median objective could be generalized to $\alpha$-percentile.

We use the following lemma from \cite{anshelevich2015approximating} in the proof of our algorithm:

\begin{lemma}
\label{lemma-social-median-alpha1}
  For any two alternatives $W$ and $Y$, we have $\alpha\text{-PC}(W) \le \alpha\text{-PC}(Y) + d(Y,W)$. [Lemma 18 in \cite{anshelevich2015approximating}]
\end{lemma}

We can generalize Lemma~\ref{lemma-social-median4} to the following lemma, and the proof is by using Lemma~\ref{lemma-social-median-alpha1} instead of Lemma~\ref{lemma-social-median1} in the proof of Lemma~\ref{lemma-social-median4},
\begin{lemma}
  \label{lemma-social-median-alpha2}
  For any three alternatives $W$, $Y$, and $P$, if $P$ pairwise defeats (or pairwise ties) $Y$, and $d(Y, W) \le d(Y, P)$, then $\alpha\text{-PC}(W) \le 3 \alpha\text{-PC}(Y)$.
\end{lemma}

\begin{theorem}
\label{thm-social-median-alpha}
The distortion of Algorithm~\ref{alg-social-median} for the $\alpha\text{-PC}$ objective social cost with $\frac{1}{2} \le \alpha \le 1$ is at most 3.
\end{theorem}

\begin{proof}
  Note that Lemma~\ref{lemma-social-median-alpha1} is actually a generalization of Lemma~\ref{lemma-social-median1}, and Lemma~\ref{lemma-social-median-alpha2} is a generalization of Lemma~\ref{lemma-social-median4}. Lemma~\ref{lemma-social-median2} and Lemma~\ref{lemma-social-median3} also generalize to the $\alpha\text{-PC}$ objective, because when $\frac{1}{2} \le \alpha \le 1$, for any alternative $Y$, we know $\alpha\text{-PC}(Y) \ge \text{med}(Y)$. Then Lemma~\ref{lemma-social-median5} also generalizes to the $\alpha\text{-PC}$ objective, because it only uses Lemma~\ref{lemma-social-median3} and Lemma~\ref{lemma-social-median4} in the proof. And Lemma~\ref{lemma-social-median6} still holds for the same algorithm. Thus all the lemmas and properties of the median objective used in the proof of Theorem~\ref{thm-social-median} could be generalized into the $\alpha\text{-PC}$ objective, so the conclusion still holds for the $\alpha\text{-PC}$ objective when $\frac{1}{2} \le \alpha \le 1$.
\end{proof}

\subsubsection{Algorithm~\ref{alg-social-median} and the Total Social Cost}
Although Algorithm~\ref{alg-social-median} is designed for the median objective, it also performs quite well for the sum objective. Interestingly, the distortion of this algorithm for the minimum total social cost is at most 5, which is the same as Copeland (the best known deterministic algorithm with no knowledge of candidate preferences). Thus this algorithm gives a distortion of 3 for median (and in fact for all $\alpha$-percentile objectives) and distortion of 5 for sum simultaneously. In settings where we are not sure which objectives to optimize, or ones where we care both about the total social good, and about fairness, this social choice mechanism provides the best of both worlds. The lemmas and proofs for this result are similar to Theorem~\ref{thm-social-median}, as follows.

\begin{lemma}
\label{lemma-social-medians-sum1}
Let $W, Y$ be alternatives $\in \mathcal{F}$. If $W$ pairwise defeats (or pairwise ties) $Y$, then $SC_{\sum}(W, \mathcal{A}) \le 3 SC_{\sum}(Y, \mathcal{A})$. [Proved in Theorem 7 in \cite{anshelevich2015approximating}]
\end{lemma}

\begin{lemma}
\label{lemma-social-median-sum2}
  For any three alternatives $W$, $Y$, and $P$, if $P$ pairwise defeats (or pairwise ties) $Y$, and $d(Y, W) \le d(Y, P)$, then $SC_{\sum}(W, \mathcal{A}) \le 5 SC_{\sum}(Y, \mathcal{A})$.
\end{lemma}

\begin {proof}
  For all $i \in \mathcal{A}$, we know $d(i, W) \le d(i, Y) + d(Y, W)$ by the triangle inequality. Summing up for all $i \in \mathcal{A}$, we get $SC_{\sum}(W, \mathcal{A}) \le SC_{\sum}(Y, \mathcal{A}) + n \cdot d(Y, W)$.

  $P$ pairwise defeats (or pairwise ties) $Y$, so at least half of the agents prefer $P$ to $Y$; thus the total social cost of $Y$ is at least the sum of the social cost of these half of agents. By Lemma~\ref{lemma-social-basic}, we get $SC_{\sum}(Y, \mathcal{A}) \ge \frac{n}{2} \frac{d(Y, P)}{2} = \frac{n}{4} d(Y, P)$. Thus,

  \begin{align*}
      SC_{\sum}(W, \mathcal{A}) &\le SC_{\sum}(Y, \mathcal{A}) + n \cdot d(Y, W)\\
                                &\le SC_{\sum}(Y, \mathcal{A}) + n \cdot d(Y, P)\\
                                &\le SC_{\sum}(Y, \mathcal{A}) + 4 SC_{\sum}(Y, \mathcal{A})\\
                                &\le 5 SC_{\sum}(Y, \mathcal{A})
  \end{align*}
\end{proof}

\begin{lemma}
\label{lemma-social-median-sum3}
  Consider the modified majority graph $G=(\mathcal{F}, E)$ at any point during Algorithm~\ref{alg-social-median}. For any edge $(W, Y) \in E$, we have that $SC_{\sum}(W, \mathcal{A}) \le 5 SC_{\sum}(Y, \mathcal{A})$.
\end{lemma}

\begin{proof}

  By Lemma~\ref{lemma-social-medians-sum1}, for any edge $(W, Y)$ in the original majority graph, $SC_{\sum}(W, \mathcal{A}) \le 3 SC_{\sum}(W, \mathcal{A})$.

  Now consider an edge $(W, Y)$ added to $E$ when processing the alternative pair $W, Y$. It must be the case that there exists an alternative $P$, such that $d(Y, W) \le d(Y, P)$ and $P$ pairwise defeats (or ties) $Y$. By Lemma~\ref{lemma-social-median-sum2}, $SC_{\sum}(W, \mathcal{A}) \le 5 SC_{\sum}(Y, \mathcal{A})$.
\end{proof}

\begin{theorem}
\label{thm-social-median-sum}
The distortion of Algorithm~\ref{alg-social-median} for minimum total social cost is at most 5, and this bound is tight.
\end{theorem}

\begin{proof}
  If there is a Condorcet winner, by Lemma~\ref{lemma-social-medians-sum1}, the distortion is at most 3.
  Otherwise, suppose the algorithm returns alternative $W$ as the winner; by Lemma~\ref{lemma-social-median-sum3} $W$ has a distortion at most 5 with any alternative $X$ as the optimal solution.

  To see that this bound is tight, consider the following example. There are three facilities $W$, $Y$, and $P$. There are $q$ agents who prefer $Y$ to $W$ to $P$, $q$ agents who prefer $P$ to $Y$ to $W$, and 1 agent who prefers $W$ to $P$ to $Y$. We denote these three sets of agents as $\mathcal{A}_Y$, $\mathcal{A}_P$ and $\mathcal{A}_W$ separately. By the preferences of agents, we know that $Y$ pairwise defeats $W$, $W$ pairwise defeats $P$, and $P$ pairwise defeats $Y$. The distances between facilities are: $d(Y, W) = 2 - 2\epsilon$, $d(W, P) = 2 - \epsilon$, $d(P, Y) = 2$, where $\epsilon$ is a very small positive number. $\mathcal{A}_Y$ is located at the same location as $Y$, so $d(\mathcal{A}_Y, Y) = 0$, $d(\mathcal{A}_Y, P)$ = 2, and $d(\mathcal{A}_Y, W) = 2 - 2\epsilon$. The distances between $\mathcal{A}_P$ and the alternatives are: $d(\mathcal{A}_P, Y) = d(\mathcal{A}_P, P) = 1$, $d(\mathcal{A}_P, W) = 3 - 2\epsilon$. $\mathcal{A}_W$ has a distance of 1 to all alternatives. Run Algorithm~\ref{alg-social-median} on this example, and consider the alternative pair $W$, $Y$. Because $P$ pairwise defeats $Y$ and $d(Y, W) \preceq d(Y, P)$, we add edge $(W, Y)$ to the graph and make $W$ the winner. The total social cost of $W$ is $q * (2-2\epsilon) + q * (3 - 2\epsilon) + 1 = q(5-4\epsilon) + 1$. While the optimal solution is to choose $Y$ as the winner, and get a total social cost of $q + 1$. When $q$ is very large and $\epsilon$ is very small, the distortion in this example approaches 5.
\end{proof}

\section{Model and Notation: Facility Assignment Problems}
\label{section-general-model}
The mechanism we used for approximation of total social cost in Theorem~\ref{thm-social-sum} can be applied to far more general problems. In this section, we describe a set of facility assignment problems that fit in this framework. As before, let $\mathcal{A} = \{ 1, 2, \dots, n\}$ be a set of agents, and $\mathcal{F} = \{ F_1, F_2, \dots, F_m\}$ be a set of facilities, with each agent $i$ having a preference ranking $\sigma_i$ over the facilities, and $\sigma = (\sigma_1, \dots, \sigma_n)$.


As in the social choice model, we assume that there exists an arbitrary unknown metric $d: (\mathcal{A} \cup \mathcal{F})^2 \rightarrow \mathbb{R}_{\ge 0}$ on the set of agents and facilities. The distances $d(i, F_j)$ between agents and facilities are unknown, but the ordinal preferences $\sigma$ and the distances $l$ between facilities are given. Let $\mathcal{D}(\sigma, l)$ be the set of metrics consistent with $\sigma$ and $l$, as defined previously in Section \ref{section-social-model}.



Unlike for social choice, our goal is now to choose which facilities to open, and which agents should be assigned to which facilities. Formally, we must choose an assignment $x:\mathcal{A}\rightarrow \mathcal{F}$, where $x(i)$ is the facility that $i$ is assigned to. Every $i \in \mathcal{A}$ must be assigned to one (and only one) facility in $\mathcal{F}$; other than that, there could be arbitrary constraints on the assignment. Here are some examples of constraints which fall into our framework: each facility $F_i$ has a capacity $c_i$, which is the maximum number of agents that can be assigned to $F_i$; at least (or at most) $p$ facilities should have agents assigned to them; agents $i$ and $j$ must be (or must not be) assigned to the same facility, etc. The social choice model is a special case of this one with the constraint that exactly one facility must be opened, and all agents must be assigned to it. Note that the constraints are only on the assignment, and independent of the metric space $d$. An assignment $x$ is valid if it satisfies all constraints.
Let $\mathcal{X}$ be the set of all valid assignments.

\paragraph{The cost function of assignments.} The cost of an assignment $x$ consists of two parts. The first part is the distance cost between agents and facilities. $\forall i \in \mathcal{A}$, let $s_i$ denote the distance between $i$ and the facility it is assigned to, i.e., $s_i=d(i,x(i))$.
For a given metric $d$ and assignment $x$, let $s(x, d)$ denote the vector of distances between each $i \in \mathcal{A}$ and $x(i)$, i.e., $s(x, d) = (s_1, s_2, \dots, s_n)$. Let $c_d: \mathbb{R}_{\ge 0}^n \rightarrow \mathbb{R}_{\ge 0}$ be a cost function that takes a vector of distances as input. For example, this could simply sum up all the distances, take the maximum distance for an egalitarian objective, etc. To be as general as possible, instead of fixing a specific function $c_d$ we consider the set of distance cost functions that are monotone nondecreasing and subadditive. Formally, $c_d$ is monotonically nondecreasing means that for any vectors $s$ and $s'$ such that $s\leq s'$ componentwise, we have that $c_d(s) \le c_d(s')$. Any reasonable cost function should satisfy this property if agents desire to be assigned to closer facilities. $c_d$ being subadditive means that for any vectors $s$ and $s'$, we have that $c_d(s+s') \le c_d(s) + c_d(s')$. While not all functions are subadditive, many important ones are, as they represent the concept of ``economies of scale", a common property of realistic costs. 


The second part of the assignment cost is the facility cost. Let $c_f(x)$ denote the facility cost for assignment $x$. $c_f$ can be an {\em arbitrary} function over the assignments, for example, the opening cost of facilities, the penalty (or reward) for assigning certain agents to the same facility, etc. Our framework includes all such functions, and thus is quite general, as we discuss below. The main components needed for our framework to work is that the function $c_f$ does not depend on the distances, only on $x$, and that the function $c_d$ is subadditive.

The total cost $c(x, d)$ of an assignment $x$ is the sum of the distance cost and the facility cost, i.e. $c(x, d) = c_d(s(x, d)) + c_f(x)$. We study algorithms to approximate the minimum cost assignment given only agents' ordinal preferences over facilities, and the distances between facilities, as described above.

\paragraph{Social Cost Distortion}
As for social choice, we use the notion of distortion to measure the quality of an assignment in the worst case, similar to the notation in \cite{boutilier2015optimal, procaccia2006distortion}. For any assignment $x$, we define the distortion of $x$ as the ratio between the social cost of $x$ and the optimal assignment:

\begin{equation*}
 dist(x, \sigma, l) = \sup\limits_{d \in \mathcal{D}(\sigma, l)} \frac{c(x, d)}{\min_{x' \in \mathcal{X}}c(x', d)}
\end{equation*}

A social choice function $f$ on $\mathcal{A}$ and $\mathcal{F}$ takes $\sigma$ and $l$ as input, and returns a valid assignment on $\mathcal{A}$ and $\mathcal{F}$. We say the distortion of $f$ on $\sigma$ and $l$ is the same as the distortion of the assignment returned by $f$. In other words, the distortion of an assignment function $f$ on a profile $\sigma$ and facility distances $l$ is the worst-case ratio between the social cost of $x = f(\sigma, l)$, and the social cost of the true optimal assignment, to obtain which we would need the true distances $d$.

\paragraph{Approximation ratio of omniscient algorithms}
Consider omniscient algorithms which know the true numerical distances between agents and facilities for the facility assignment problems, in other words, the metric $d$. In some sense, the goal of our work is to determine when algorithms with only limited information can compete with such omniscient algorithms. With the full distances information, we can of course obtain the optimal assignment using brute force, while for our algorithms with limited knowledge this is impossible even given unlimited computational resources. Nevertheless, we are also interested in what is possible to achieve if we restrict ourselves to polynomial time. To differentiate traditional approximation algorithms from algorithms with small distortion, suppose that an omniscient approximation algorithm $\tilde{f}$ returns assignment $x$. Then we denote the approximation ratio of a valid assignment $x$ as:


\begin{equation*}
  ratio(x) = \frac{c(x, d)}{\min_{x' \in \mathcal{X}}c(x', d)}
\end{equation*}

Thus we say the approximation ratio of an omniscient algorithm $\tilde{f}$ is $\beta$ if for any input of the problem, the assignment $x$ returned by $\tilde{f}$ has $ratio(x) \le \beta$.

\subsection{Examples of Facility Assignment Problems}
\label{subsection-examples}
In this section we illustrate that our framework is quite general by giving various important examples which fit into our framework. In the section which follows, we prove a general black-box reduction theorem for our framework, and thus immediately obtain mechanisms with small distortion for all these examples simultaneously.

The total social cost problem we discussed in Section~\ref{subsection-social-sum} is a special case of the facility assignment problem such that the constraint is only one facility (alternative) is chosen, and all agents are assigned to it. 
For any assignment $x$, the facility cost function $c_f(x) = 0$, and the distance cost function $c_d(s(x, d))$ is the sum of distances from the winning alternative to all agents in the metric $d$. $c_d$ is monotone and additive (thus subadditive). Here are some other examples that fit in our framework:\\

\textbf{Minimum weight metric bipartite matching.}
Given a set of agents $\mathcal{A}$ and a set of facilities $\mathcal{F}$ such that $|\mathcal{A}| = |\mathcal{F}| = n$. $G = (\mathcal{A}, \mathcal{F}, E)$ is an undirected complete bipartite graph. The facilities and agents lie in a metric space $d$. The weight of each edge $(i, F) \in E$ is the distance between $i$ and $F$, $w(i, F) = d(i, F)$. The goal is to find a minimum weight perfect matching of the bipartite graph given only ordinal information. This setting has been studied before, and the best distortion bound known is $n$ \cite{caragiannis2016truthful} given by RSD for the case when only the ordinal preferences $\sigma$ are known. Our results show that if we also know the distances $l$ between facilities, then even without knowing the distances $d$ between agents and facilities, it is possible to create simple mechanisms with distortion at most 3 (we can show that no better bound is possible for this setting). Thus having a bit more information about the facilities immediately improves the distortion bound by a very large amount. We show this result by using our facility assignment framework above: the constraint here is that each facility has a capacity of 1, thus a valid assignment is a perfect matching of the bipartite graph. For any assignment $x$, the facility cost function is $c_f(x) = 0$, and the distance cost function $c_d(s(x, d))$ is the total edge weight in the assignment. $c_d$ is monotone and additive (thus subadditive).\\

\textbf{Egalitarian bipartite matching.}
With the same bipartite graph as in minimum weight matching problems, the only difference is that the goal of egalitarian bipartite matching is to find a perfect matching such that maximum edge weight (instead of the total weight) in the matching is minimized \cite{bogomolnaia2004random}.

The egalitarian bipartite matching problem is the same as minimum weight bipartite matching except the distance cost function $c_d(s(x, d))$ is the maximum edge weight in the assignment. This function is also monotone and subadditive.\\

\textbf{Metric Facility Location.}
In this problem, one is given a set of agents $\mathcal{A}$ and a set of facilities $\mathcal{F}$ such that $|\mathcal{A}| = n$,  $|\mathcal{F}| = m$. The facilities and agents lie in a metric space $d$. Each facility $F_j \in \mathcal{F}$ has an opening cost $f_j$. Each agent is assigned to a facility; in different versions there may be capacities on the number of agents assigned to a facility, lower bounds on the number of agents assigned to a facility, or various other constraints \cite{farahani2009facility}. The goal is to find a subset of facilities  $\mathcal{\hat{F}} \subseteq \mathcal{F}$ to open, so that the sum of opening costs for facilities in $\mathcal{\hat{F}}$ and total distance of the assignment is minimized.

Our framework allows arbitrary constraints on what constitutes a valid assignment, which captures facilities with capacities or lower bounds if needed. For any assignment $x$, the facility cost function $c_f(x)$ is the sum of the opening costs $f_j$ for those facilities $F_j$ that have at least one agent assigned to it. The distance cost function $c_d(s(x, d))$ is the total distances in the assignment, which is monotone increasing and additive (thus subadditive).\\

\textbf{$k$-center problem.}
The goal in this classic problem is to open a set of $k$ facilities, with each agent assigned to the closest one. The optimal solution is the subset of $\mathcal{\hat{F}}$ which minimizes $\max_{i \in \mathcal{A}} d(i, x(i))$. To express this in our framework, the constraint is that no more than $k$ facilities have agents assigned to them. For any assignment $x$, the facility cost function $c_f(x) = 0$, and the distance cost function $c_d(s(x, d))$ is the maximum distance between any agent and facility in the assignment.\\

\textbf{$k$-median problem.}
This classic problem is the same as $k$-center, except the goal is to minimize the sum of distances of agents to the facilities instead of the maximum distance.


\section{Distortion of Facility Assignment Problems}
\label{section-facility-assignment}

In this section, we study general facility assignment problems, as described in Section \ref{section-general-model}, and form mechanisms with small distortion. First, we construct a projected problem such that the distances between agents and facilities are known, so it could be solved by an omniscient algorithm. Then we map the result of the projected problem to the original problem and bound the distortion of the original problem.

Given agents $\mathcal{A} = \{ 1, 2, \dots, n\}$ and facilities $\mathcal{F} = \{ F_1, F_2, \dots, F_m\}$, suppose facility $F'$ is $i$'s top choice in $\mathcal{F}$. We create a new agent $\tilde{i}$ at the location of $F'$ in the metric space. Consequently, $\forall F \in \mathcal{F}$, $d(\tilde{i}, F) = d(F', F)$. Denote the set of the new agents as $\tilde{\mathcal{A}} = \{\tilde{1}, \tilde{2}, \dots, \tilde{n} \}$.

The original assignment problem is on agents $\mathcal{A}$ and facilities $\mathcal{F}$, and only ordinal preferences of agents in $\mathcal{A}$ over facilities are given. The projected problem is on agents $\tilde{\mathcal{A}}$ and facilities $\mathcal{F}$, and we know the actual distances between agents in $\tilde{\mathcal{A}}$ and facilities $\mathcal{F}$, since we know the distances $l$ between facilities. The constraints and costs $c_d$ and $c_f$ remain the same for both the original and the projected problem; the only difference is in the distances $d$. Our main result is that if we have a $\beta$-approximation assignment to the minimum assignment cost on the projected problem, then we can get an assignment that has a distortion of $2\beta + 1$ for the original problem in polynomial time.\\

%
%

\begin{theorem}
\label{thm-facility-projection}
  Given a valid assignment $\tilde{x}$ for the projected problem on $\tilde{\mathcal{A}}$ and $\mathcal{F}$, with $ratio(\tilde{x}) \le \beta$, the assignment $x(i)=\tilde{x}(\tilde{i})$ has distortion of at most $(1 + 2\beta)$ for original assignment problem on $\mathcal{A}$ and $\mathcal{F}$.
\end{theorem}

\begin{proof}
  First, $\tilde{x}$ is a valid assignment for the projected problem on $\tilde{\mathcal{A}}$ and $\mathcal{F}$, so $x$ must also be a valid assignment for the original problem on $\mathcal{A}$ and $\mathcal{F}$. This is because the constraints are only on the assignment, and are independent of the metric space $d$. For the same reason, the facility cost of $x$ equals the facility cost of $\tilde{x}$, $c_f(x) = c_f(\tilde{x})$.

  Now consider the distance cost of $x$. Let $x^*$ denote the optimal assignment for the original problem. $\forall i \in \mathcal{A}$, let $s_i = d(i, x(i))$, $t_i = d(i, \tilde{i})$, $b_i = d(\tilde{i}, x(i))$. Similarly, let $s_i^* = d(i, x^*(i))$, $b_i^* = d(\tilde{i}, x^*(i))$.

  For any agent $i$ and facility $x(i)$, by triangle inequality,

  \begin{align*}
    s_i = d(i, x(i)) &\le d(i, \tilde{i}) + d(\tilde{i}, x(i)) = t_i + b_i
  \end{align*}

  Because $c_d$ is monotonically nondecreasing and subadditive,

  \begin{align*}
    c_d(s_1, s_2, \dots, s_n) &\le c_d(t_1 + b_1, t_2 + b_2, \dots, t_n + b_n) \\
                              &\le c_d(t_1, t_2, \dots, t_n) + c_d(b_1, b_2, \dots, b_n)
  \end{align*}

  Therefore, the cost of our assignment $x$ is bounded as follows:
  \begin{align*}
    c_f(x) + c_d(s(x, d)) &= c_f(x) + c_d(s_1, s_2, \dots, s_n) \\
                              &= c_f(\tilde{x}) + c_d(s_1, s_2, \dots, s_n) \\
                              &\le c_f(\tilde{x}) + c_d(t_1, t_2, \dots, t_n) + c_d(b_1, b_2, \dots, b_n)
  \end{align*}

  Because $\tilde{i}$ is located at $i$'s top choice facility, and $x^*(i)$ is a facility, we thus know that $t_i \le s_i^*$, and by monotonicity $c_d(t_1, t_2, \dots, t_n) \le c_d(s_1^*, s_2^*, \dots, s_n^*)$. Thus,

  \begin{align*}
    c_f(x) + c_d(s(x, d)) &\le c_f(\tilde{x}) + c_d(t_1, t_2, \dots, t_n) + c_d(b_1, b_2, \dots, b_n)\\
                               &\le c_f(\tilde{x}) + c_d(s_1^*, s_2^*, \dots, s_n^*) + c_d(b_1, b_2, \dots, b_n)
  \end{align*}

  We know that $\tilde{x}$ is a $\beta$-approximation to the optimum assignment for the projected problem. Its total cost is exactly $c_f(\tilde{x}) + c_d(b_1, b_2, \dots, b_n)$, since the distance from $\tilde{i}$ to $\tilde{x}(\tilde{i})=x(i)$ is exactly $b_i$. Now consider another assignment for the projected problem, in which $\tilde{i}$ is assigned to $x^*(i)$. The cost of this assignment is $c_f(x^*) + c_d(b_1^*, b_2^*, \dots, b_n^*)$, by definition of $b_i^*$. Since $\tilde{x}$ is a $\beta$-approximation, we therefore know that

  \begin{align*}
  c_f(\tilde{x}) + c_d(b_1, b_2, \dots, b_n) \leq \beta c_f(x^*) + \beta c_d(b_1^*, b_2^*, \dots, b_n^*),
  \end{align*}

  and thus

   \begin{align*}
    c_f(x) + c_d(s(x, d)) &\le c_f(\tilde{x}) + c_d(s_1^*, s_2^*, \dots, s_n^*) + c_d(b_1, b_2, \dots, b_n)\\
                    &\le c_d(s_1^*, s_2^*, \dots, s_n^*) + \beta c_f(x^*) + \beta c_d(b_1^*, b_2^*, \dots, b_n^*)
  \end{align*}

  For any agent $i$ and facility $x^*(i)$ in $x^*$, by triangle inequality,

  \begin{align*}
     b_i^* = d(\tilde{i}, x^*(i)) \leq d(i, x^*(i)) + d(i, \tilde{i}) \leq 2 d(i, x^*(i)) = 2s_i^*
  \end{align*}
  $d(i, \tilde{i})\leq d(i, x^*(i))$ above since $\tilde{i}$ is located at the closest facility to $i$. Because $c_d$ is monotone and subadditive, we also have that

  \begin{align*}
  c_d(b_1^*, b_2^*, \dots, b_n^*) \leq c_d(2s_1^*, 2s_2^*, \dots, 2s_n^*) \leq 2 c_d(s_1^*, s_2^*, \dots, s_n^*)
  \end{align*}




  Putting everything together,

  \begin{align*}
    c_f(x) + c_d(s(x, d)) 
                              &\le c_d(s_1^*, s_2^*, \dots, s_n^*) + \beta c_f(x^*) + \beta c_d(b_1^*, b_2^*, \dots, b_n^*)\\
                              &\le \beta c_f(x^*) + c_d(s_1^*, s_2^*, \dots, s_n^*) + 2\beta c_d(s_1^*, s_2^*, \dots, s_n^*)\\
                              &= \beta c_f(x^*) + (1 + 2\beta) c_d(s_1^*, s_2^*, \dots, s_n^*)\\
                              &\le (1 + 2\beta) (c_f(x^*) + c_d(s(x^*, d)))
  \end{align*}

\end{proof}


Note that the above theorem immediately implies that if we are only concerned with what is possible to achieve given limited ordinal information in addition to distances between facilities, and are not worried about our algorithms running in polynomial time, then we can always form an assignment with distortion of at most 3 from knowing only $\sigma$ and $l$. This is because we can solve the projected problem with brute force, and then we have $\beta=1$. This bound of 3 is tight for many facility assignment problems: consider for example an instance of min-cost metric matching with two agents and two facilities, with both preferring $F_1$ to $F_2$. One of the agents has distance to $F_1$ of 0, and one is located halfway between $F_1$ and $F_2$, but since we only have ordinal information we do not know which agent is which. If we assign the wrong agent to $F_1$, then we end up with distortion of 3, and it is impossible to do better for any deterministic mechanism.

If on the other hand we want to form poly-time algorithms with small distortion, the above theorem gives a black-box reduction: if we have a $\beta$-approximation algorithm for the omniscient case, then we can form a $1+2\beta$-distortion algorithm for the ordinal case. Actually, we get a $1+2\beta$-distortion for the distance cost, and a $\beta$-distortion for the facility cost, which is shown in the second-to-last line of the proof for Theorem~\ref{thm-facility-projection}. This leads to the following corollaries:

\begin{corollary}
We can achieve the following distortion in polynomial time:
\begin{enumerate}
\item At most 3 for the minimum weight bipartite matching problem.
\item At most 3 for Egalitarian bipartite matching.
\item At most 3.976 for the facility location problem (1.488-approximation for the facility cost, and 3.976-approximation for the distance cost).
\item At most 5 for the k-center problem.
\item At most 6.35 for the k-median problem.
\end{enumerate}
\end{corollary}

\begin{proof}
Min-cost matching and egalitarian matching are poly-time solvable, so $\beta=1$. For the latter, one can fix the threshold weight $t$ such that every edge chosen should be at most $t$, and then determine if such a matching exists. Performing a binary search on $t$ gives an efficient algorithm. For facility location, one can use the omniscient algorithm which is a 1.488-approximation in \cite{li20111}. For the $k$-center problem, a greedy algorithm \cite{hochbaum1985best} gives a 2-approximation for the setting that agents are a subset of facilities, which is the case in our projected problem. \cite{byrka2014improved} gives a 2.675-approximation omniscient algorithm for the $k$-median problem when agents are a subset of facilities, thus it also gives a 2.675-approximation for our projected problem.
\end{proof}

Note that the median function, unlike sum and maximum, is not subadditive, and thus does not fit into our framework. In fact, while both min-cost and egalitarian matching problems have algorithms with small distortion in our setting, the same is not possible for forming a matching where the objective function is the cost of the {\em median} edge: see the Appendix for a lower bound.

\section{Conclusion}
In this paper, we provided two mechanisms to solve different social cost problems. The first one makes use of the distances between facilities and an omniscient algorithm to get a low distortion for general facility assignment problems. The second mechanism is a new voting rule for social choice which simultaneously achieves a distortion of 3 for many objectives, including the cost of the median voter, and a distortion of 5 for the total social cost at the same time. The first mechanism requires the full distances $l$, but only needs the top choice from each agent. Thus, it puts only a small load on the agents which submit their preferences, but requires the mechanism designer to collect more information about the facilities and their distances to each other. The second mechanism, on the other hand, only requires ordinal preference information from the facilities, but needs the full preference ranking from the agents instead of just the top choice. It is especially appropriate for settings in which the candidates or alternatives are agents themselves.

Many open questions remain for our setting. How well can facility location problems be approximated given information about facilities? While we established upper bounds on distortion, we have no lower bounds besides the trivial bound of 3. What about randomized mechanisms, or what if the mechanisms must be truthful? And more generally, exactly what information is enough to guarantee mechanisms with small distortion? While our results show that knowing information about facility locations is enough to result in small distortion, it may be possible that obtaining even a bit of targeted information would result in powerful approximation algorithms. We look forward to future work on this topic.

\bibliographystyle{plain}
\bibliography{ref}

\begin{thebibliography}{10}

\bibitem{abramowitz2017utilitarians}
Ben Abramowitz and Elliot Anshelevich.
\newblock Utilitarians without utilities: Maximizing social welfare for graph
  problems using only ordinal preferences.
\newblock In {\em AAAI 2018}.

\bibitem{anshelevich2016ordinal}
Elliot Anshelevich.
\newblock Ordinal approximation in matching and social choice.
\newblock {\em ACM SIGecom Exchanges}, 15(1):60--64, 2016.

\bibitem{anshelevich2015approximating}
Elliot Anshelevich, Onkar Bhardwaj, and John Postl.
\newblock Approximating optimal social choice under metric preferences.
\newblock In {\em AAAI 2015}.

\bibitem{anshelevich2015randomized}
Elliot Anshelevich and John Postl.
\newblock Randomized social choice functions under metric preferences.
\newblock In {\em IJCAI 2016}.

\bibitem{anshelevich2015blind}
Elliot Anshelevich and Shreyas Sekar.
\newblock Blind, greedy, and random: Algorithms for matching and clustering
  using only ordinal information.
\newblock In {\em AAAI 2016}.

\bibitem{anshelevich2016truthful}
Elliot Anshelevich and Shreyas Sekar.
\newblock Truthful mechanisms for matching and clustering in an ordinal world.
\newblock In {\em WINE 2016}.

\bibitem{anshelevich2017tradeoffs}
Elliot Anshelevich and Wennan Zhu.
\newblock Tradeoffs between information and ordinal approximation for bipartite
  matching.
\newblock In {\em SAGT 2017}.

\bibitem{benade2017preference}
Gerdus Benade, Swaprava Nath, Ariel~D Procaccia, and Nisarg Shah.
\newblock Preference elicitation for participatory budgeting.
\newblock In {\em AAAI 2017}.

\bibitem{bhalgat2011social}
Anand Bhalgat, Deeparnab Chakrabarty, and Sanjeev Khanna.
\newblock Social welfare in one-sided matching markets without money.
\newblock In {\em APPROX 2011}.

\bibitem{bogomolnaia2004random}
Anna Bogomolnaia and Herv{\'e} Moulin.
\newblock Random matching under dichotomous preferences.
\newblock {\em Econometrica}, 72(1):257--279, 2004.

\bibitem{boutilier2015optimal}
Craig Boutilier, Ioannis Caragiannis, Simi Haber, Tyler Lu, Ariel~D Procaccia,
  and Or~Sheffet.
\newblock Optimal social choice functions: A utilitarian view.
\newblock {\em Artificial Intelligence}, 227:190--213, 2015.

\bibitem{byrka2014improved}
Jaros{\l}aw Byrka, Thomas Pensyl, Bartosz Rybicki, Aravind Srinivasan, and Khoa
  Trinh.
\newblock An improved approximation for k-median, and positive correlation in
  budgeted optimization.
\newblock In {\em SODA 2014}.

\bibitem{caragiannis2016truthful}
Ioannis Caragiannis, Aris Filos-Ratsikas, S{\o}ren Kristoffer~Stiil
  Frederiksen, Kristoffer~Arnsfelt Hansen, and Zihan Tan.
\newblock Truthful facility assignment with resource augmentation: An exact
  analysis of serial dictatorship.
\newblock In {\em WINE 2016}.

\bibitem{caragiannis2017subset}
Ioannis Caragiannis, Swaprava Nath, Ariel~D Procaccia, and Nisarg Shah.
\newblock Subset selection via implicit utilitarian voting.
\newblock {\em Journal of Artificial Intelligence Research}, 58:123--152, 2017.

\bibitem{cheng2017people}
Yu~Cheng, Shaddin Dughmi, and David Kempe.
\newblock Of the people: voting is more effective with representative
  candidates.
\newblock In {\em EC 2017}.

\bibitem{cheng2018distortion}
Yu~Cheng, Shaddin Dughmi, and David Kempe.
\newblock On the distortion of voting with multiple representative candidates.
\newblock In {\em AAAI 2018}.

\bibitem{christodoulou2016social}
George Christodoulou, Aris Filos-Ratsikas, S{\o}ren Kristoffer~Stiil
  Frederiksen, Paul~W Goldberg, Jie Zhang, and Jinshan Zhang.
\newblock Social welfare in one-sided matching mechanisms.
\newblock In {\em AAMAS 2016}.

\bibitem{enelow1984spatial}
James~M. Enelow and Melvin~J. Hinich.
\newblock {\em The Spatial Theory of Voting: An Introduction}.
\newblock Cambridge University Press, New York, NY, 1984.

\bibitem{fain2017sequential}
Brandon Fain, Ashish Goel, Kamesh Munagala, and Sukolsak Sakshuwong.
\newblock Sequential deliberation for social choice.
\newblock In {\em WINE 2017}.

\bibitem{farahani2009facility}
Reza~Zanjirani Farahani and Masoud Hekmatfar.
\newblock {\em Facility location: concepts, models, algorithms and case
  studies}.
\newblock Springer, 2009.

\bibitem{feldman2016voting}
Michal Feldman, Amos Fiat, and Iddan Golomb.
\newblock On voting and facility location.
\newblock In {\em EC 2016}.

\bibitem{filos2014social}
Aris Filos-Ratsikas, S{\o}ren Kristoffer~Stiil Frederiksen, and Jie Zhang.
\newblock Social welfare in one-sided matchings: Random priority and beyond.
\newblock In {\em SAGT 2014}.

\bibitem{goel2016metric}
Ashish Goel, Anilesh~Kollagunta Krishnaswamy, and Kamesh Munagala.
\newblock Metric distortion of social choice rules: Lower bounds and fairness
  properties.
\newblock In {\em EC 2017}.

\bibitem{gross2017vote}
Stephen Gross, Elliot Anshelevich, and Lirong Xia.
\newblock Vote until two of you agree: Mechanisms with small distortion and
  sample complexity.
\newblock In {\em AAAI 2017}.

\bibitem{hochbaum1985best}
Dorit~S Hochbaum and David~B Shmoys.
\newblock A best possible heuristic for the k-center problem.
\newblock {\em Mathematics of operations research}, 10(2):180--184, 1985.

\bibitem{hoefer2017combinatorial}
Martin Hoefer and Bojana Kodric.
\newblock Combinatorial secretary problems with ordinal information.
\newblock In {\em ICALP 2017}.

\bibitem{li20111}
Shi Li.
\newblock A 1.488 approximation algorithm for the uncapacitated facility
  location problem.
\newblock In {\em ICALP 2011}.

\bibitem{procaccia2006distortion}
Ariel~D Procaccia and Jeffrey~S Rosenschein.
\newblock The distortion of cardinal preferences in voting.
\newblock In {\em CIA 2006}.

\bibitem{skowron2017social}
Piotr~Krzysztof Skowron and Edith Elkind.
\newblock Social choice under metric preferences: Scoring rules and stv.
\newblock In {\em AAAI 2017}.

\end{thebibliography}

\newpage
\appendix
\section{Bad Examples and Lower Bounds}

Note that our Algorithm \ref{alg-social-median} is only for social choice problems, and does not fit in the definition of our general facility assignment problems. This is because the median cost function, unlike sum and maximum, is not subadditive. In fact, while both min-cost and egalitarian matching problems have algorithms with small distortion in our setting, the same is not possible for forming a matching where the objective function is the cost of the {\em median} edge.

\begin{theorem}
The worst-case distortion of the median-cost bipartite matching problem in a metric space (given both agent preference profiles and distances between facilities) is unbounded.
\end{theorem}

\begin{proof}
Consider the following example: there are three agents $a$, $b$, $c$, and three facilities $X$, $Y$, $Z$. The preferences of agents are: $a, b \in XYZ$, while $c \in ZXY$. The distances between facilities are: $l(X, Y) = 2$, $l(X, Z) = l(Y, Z) = 1000$. The distances between the agents and facilities are, of course, unknown. Consider the instance $d(c, Z) = \epsilon$, $d(a, X) = 2\epsilon$, and $d(b, X) = d(b, Y) = 1$. $\epsilon$ is a very small positive real number, and other distances not given obey triangle inequality. In this instance, the optimal solution is $x^* = \{ (a, X), (b, Y), (c, Z) \}$, which gives a median value of $2\epsilon$. But because $a$ and $b$ have the same preference profile, the instance could also be $d(c, Z) = \epsilon$, $d(b, X) = 2\epsilon$, and $d(a, X) = d(a, Y) = 1$. If we still return the assignment $x^*$ for this instance, the median would be $1$. The distortion is arbitrarily bad when $\epsilon$ approaches 0.
\end{proof}

The following Theorems show some of the lower bounds mentioned in Table~\ref{table_results}.

\begin{theorem}
The worst-case distortion for the facility location problem in a metric space (given only agents' preference profiles) is unbounded.
\end{theorem}

\begin{proof}
Consider the following example: there are two agents 1, 2, and two facilities $X$, $Y$. Agent 1 prefers $X$ to $Y$, while agent 2 prefers $Y$ to $X$. The opening costs are: $c_f(X) = 1$, $c_f(Y) = 100$. We can choose to open one facility or both of them.

\textbf{Case 1.} Suppose we only open $X$. Consider the following distances between the agents and facilities: $d(1, X) = d(2, Y) = 1$, $d(1, Y) = d(2, X) = L$, for some very large $L$. If we only open $X$, then the total cost is $>L$. While the optimal solution is to open both $X$ and $Y$, which has a total cost of 103. The distortion is unbounded.

\textbf{Case 2.} Suppose we only open $Y$. Consider the same distances as in \textbf{Case 1}, then the total cost is also $L$. And the optimal solution still has a total cost of 103. The distortion is unbounded.

\textbf{Case 3.} Suppose we open both facilities. Consider the following distances between the agents and facilities: $d(1, X) = d(1, Y) = d(2, X) = d(2, Y) = \epsilon$, where $\epsilon$ is a very small positive real number. If we open both facilities, the total cost is $101 + 2\epsilon$. While the optimal solution is to only open $X$ , which has a total cost of $1 + 2\epsilon$. If we increase $c_f(Y)$, the approximation ratio is unbounded.
\end{proof}

\begin{theorem}
The worst-case distortion for the k-median problem in a metric space (given only agents' preference profiles) is at least $\Omega(n)$.
\end{theorem}

\begin{proof}
Consider the following example: There are three facilities $X$, $Y$, and $Z$. There are $q$ agents who prefer $X$ to $Y$ to $Z$, $q$ agents who prefer $Y$ to $X$ to $Z$, and 1 agent who prefers $Z$ to $X$ to $Y$. We denote these three sets of agents as $\mathcal{A}_X$, $\mathcal{A}_Y$ and $\mathcal{A}_Z$ separately. Suppose $k=2$, then we have three choices of the winners:

\textbf{Case 1.} Choose $X, Y$ as the winners. Consider the following distances between agents and facilities: $d(X, Y) = 1$, $d(Y, Z) = d(X, Z) = L$ for some very large $L$. $\mathcal{A}_X$ is located at the same location as $X$, $\mathcal{A}_Y$ is at the same location as $Y$, and $\mathcal{A}_Z$ is at the same location as $Z$. The cost of choosing $X, Y$ as the winners is $L$ because we need to assign the agent in $\mathcal{A}_Z$ to $X$ or $Y$. While the optimal solution is to choose $Y, Z$ as the winners, and get a total cost of $q$. So the distortion in this case is unbounded.

\textbf{Case 2.} Choose $X, Z$ as the winners. Consider the following distances between agents and facilities: $d(X, Y) = d(Y, Z) = d(X, Z) = 1$, and $\mathcal{A}_X$ locate on top of $X$, $\mathcal{A}_Y$ locate on top of $Y$, and $\mathcal{A}_Z$ locate on top of $Z$. The cost of choosing $X, Z$ as the winner is $q$, while the optimal solution is to choose $X, Y$ as the winners, and get a total cost of $1$. The distortion is $q$ in this case.

\textbf{Case 3.} Choose $Y, Z$ as the winners. Consider the same distances as in \textbf{Case 2}. If we choose $Y, Z$ as the winners, the total cost is still $q$, and the distortion of this case is also $q$.

The total number of agents is $n = 2q + 1$, so we can conclude that the distortions in all these three cases are at least $\Omega(n)$.
\end{proof}

\begin{theorem}
The worst-case distortion of the egalitarian bipartite matching problem in a metric space (given only agents' preference profiles) is at least $2$.
\end{theorem}

\begin{proof}
Consider the following example: there are two agents 1, 2, and two facilities $X$, $Y$. Both agents prefer $X$ to $Y$. W.L.O.G., assume we match agent 1 to $X$, and agent 2 to $Y$. Suppose the distances between agents and facilities are: $d(1, X) = d(1, Y) = 1$, $d(2, X) = \epsilon$, $d(2, Y) = 2$, where $\epsilon$ is a very small positive real number. The egalitarian cost of our matching is $2$, while the optimal solution is to match agent 1 to $Y$, and agent 2 to $X$, which has a cost of $1$.
\end{proof}
\end{document}